\documentclass[12pt]{article}
\usepackage{amsmath, amsfonts, amssymb, amsthm, geometry, graphicx, arydshln, umoline, subfig, newtxtext, newtxmath, color, comment, soul, enumerate, bm,float,multirow,booktabs,varwidth}
\usepackage{tikz}
\usetikzlibrary{matrix, positioning,arrows.meta}
\usepackage{appendix}
\usepackage[colorlinks=true,citecolor=blue]{hyperref}
\usepackage{natbib}
\usepackage{booktabs}
\usepackage[table]{xcolor}
\def\BibTeX{{\rm B\kern-.05em{\sc i\kern-.025em b}\kern-.08em
    T\kern-.1667em\lower.7ex\hbox{E}\kern-.125emX}}
\usepackage{makecell}
\usepackage{bbm}

\numberwithin{equation}{section} %%%G (added)

\newcommand{\Delete}[1]{}   %%%F (added)

\theoremstyle{definition}
\newtheorem{theorem}{Theorem}%[section]

\newtheorem{corollary}{Corollary}

\newtheorem{example}{Example}

\newtheorem{lemma}{Lemma}

\newtheorem{proposition}{Proposition}

\begin{document}

\title{A New Method for Finding the Schulze Winner Set}
\author{
Satoru Fujishige\thanks{Research Institute for Mathematical Sciences, Kyoto University, Kyoto 606-8502, Japan. Email: fujishig@kurims.kyoto-u.ac.jp}
%%F \and
\qquad
Leo Goto\thanks{Undergraduate School of Management, Department of Business Economics, Tokyo University of Science, 1-11-2, Fujimi, Chiyoda-ku, Tokyo, 102-0071, Japan. Email: leogotodeb@gmail.com}
%%F \and
\qquad
Satoshi Nakada\thanks{School of Management, Department of Business Economics, Tokyo University of Science, 1-11-2, Fujimi, Chiyoda-ku, Tokyo, 102-0071, Japan. Email: snakada@rs.tus.ac.jp}
}
%\date{\today}
\date{June 4, 2026}
\maketitle

\begin{abstract}
We propose a new voting algorithm based on the pairwise majority-comparison matrix derived from voters' preference profiles. 
We show that this algorithm induces exactly the winner set of the Schulze rule \citep{Schulze1997}. Our algorithm successively eliminates 
%%F the weakest candidates in terms of 
weaker candidates in terms of 
all-pairs comparisons, thereby reflecting a dual spirit to Condorcet's original idea of splitting preference cycles \citep{Condorcet1785}. 
We further show that the direct sum of 
%%F the winner sets 
the survival sets 
obtained at each elimination round coincides with the Schwartz set \citep{Schwartz1972}. 
These two equivalence results provide a formal mathematical foundation for the ``folklore'' relationship between the Schulze winner set and the Schwartz set, as well as a new Condorcetian interpretation of the Schulze winner set.

\end{abstract}

\section{Introduction}\label{sec:intro}

The problem of aggregating individuals' preferences to determine a socially desirable alternative has been studied for centuries and constitutes one of the central themes of social choice theory. 
When there are only two alternatives, the problem is comparatively simple: the majority rule is widely regarded as the most natural democratic procedure, and its axiomatic characterization was established by \citet{May1952}. 
However, once there are three or more alternatives, pairwise majority preferences may fail to satisfy transitivity, thereby giving rise to the phenomenon known as the \textit{Condorcet cycle}.

To address this difficulty, \citet{Condorcet1785} proposed a remarkable procedure based on pairwise majority comparisons. 
He first considered the complete system of pairwise majority propositions and observed that this system may become ``inconsistent'' due to cyclic majority relations. 
Condorcet then suggested successively eliminating those propositions supported by the weakest pluralities until a consistent system remains.\footnote{\citet{Condorcet1785} describes this procedure as follows:
\begin{quote}
One would obtain in this way the system of propositions formed by the plurality among the possible systems. [...]
If the system formed by the propositions which have obtained the plurality is one of the possible systems, the candidate whom this system places first should be regarded as elected. 
If, on the contrary, this system is one of the impossible systems, one must successively reject from this impossible system the propositions which have the smaller plurality, and accept the resulting system of the remaining propositions.
\citep[pp.~lx--lxi, lxvii--lxix]{Condorcet1785}
\end{quote}
}
This idea naturally leads to the alternative, \textit{Condorcet winner}, namely an alternative that defeats every other alternative in pairwise majority comparison. 
A voting rule that always selects the Condorcet winner whenever it exists is called a \textit{Condorcet-consistent rule}.

Among the well-used Condorcet-consistent rules is the \textit{Schulze rule}, introduced by \citet{Schulze1997}.\footnote{The rule has been adopted in practice by several organizations and open-source communities, including the Debian Project, Software in the Public Interest, and Wikimedia-related elections.} 
The Schulze rule evaluates how strongly an alternative $x$ defeats another alternative $y$ through paths in the pairwise majority graph. 
More precisely, for every path from $x$ to $y$, one considers the weakest pairwise support of alternative $z$ against another alternative $w$ along that path, and the strength of $x$ against $y$ is defined as the maximum such value over all paths from $x$ to $y$. 
The winner is then determined according to these strongest-path comparisons. 
From a graph-theoretic viewpoint, computing the Schulze winner set is an instance of the all-pairs bottleneck path problem, which can be solved by classical algorithms such as the Floyd--Warshall algorithm \citep{Floyd1962, Warshall1962}.

We propose a new algorithm inspired by the original spirit of Condorcet's successive elimination procedure. 
Our algorithm also relies on pairwise majority comparisons, but instead of propagating strongest paths, it successively eliminates 
%%F the weakest pairwise relations between alternatives. 
weaker pairwise relations between alternatives. 
Intuitively, if an alternative $x$ is insufficiently supported against some alternative, then the possibility that $x$ should ultimately win is discarded. 
While Condorcet's original proposal successively eliminated pairwise relations with the smallest plurality margins, our algorithm successively eliminates the weaker pairwise dominance relations induced by the majority graph. 
We call this algorithm \textit{SuccessiveDicut}.

Our main result shows that 
%%F the winner set 
the survival set  %%F
produced by SuccessiveDicut coincides exactly with the Schulze winner set. 
This equivalence provides a new interpretation of the Schulze rule from the dual perspective of Condorcet's original idea of successively eliminating weak majority defeats.
Furthermore, each stage of our elimination procedure induces an equivalence class of alternatives that reflects the hierarchical structure of the underlying pairwise majority relation. 
We show that the direct sum of these equivalence classes in each round coincides with the Schwartz set \citep{Schwartz1972}. 
The relationship between the Schulze rule and the Schwartz set has long been informally discussed in the literature, and \citet{Schulze2003} explains this connection through the so-called Schwartz set heuristic. 
However, to the best of our knowledge, a formal structural characterization has not previously been established. 
Therefore, our results provide a mathematical foundation for the well-known folklore relationship between the Schulze winner set and the Schwartz set.

The rest of this paper is organized as follows.
In Section~\ref{sec: Schulze}, we introduce the basic setup and review the definition of the Schulze rule. 
In Section~\ref{sec:successiveDicut}, we present our algorithm and establish the main equivalence results. 
Section~\ref{sec: KM} discusses the relationship between our method and the Kemeny--Young method.
Further discussions are provided in Section~\ref{sec:discussion}. 
Omitted proofs are collected in the appendices.

\subsection*{Related literature}
Our paper contributes to the literature on designing Condorcet-consistent voting rules. 
One influential interpretation of Condorcet's proposal is due to \citet{Young1988}, who gives a maximum-likelihood interpretation of Condorcet's method and connects it to the problem of identifying the ranking most likely to reflect the voters' collective judgment. 
He further shows that this maximum-likelihood approach coincides with the rule proposed by \citet{Kemeny1959}, in which the social ranking is determined by minimizing the sum of Kendall-tau distances between the social ranking and each individual's preference.\footnote{See also \citet{young1978consistent} for an axiomatic characterization of the Kemeny--Young rule.}

In contrast to \citet{Young1988}'s interpretation, our approach reconstructs Condorcet's proposal as a procedure that iteratively eliminates weak majority relations until a stable acyclic structure remains. 
Mathematically, whereas the Kemeny--Young rule can be formulated as an instance of the linear ordering problem on a weighted graph, our algorithm provides a complementary perspective based on the all-pairs bottleneck path problem.

This viewpoint is closely related to the notion of the Schwartz set introduced by \citet{Schwartz1972}, which was proposed as a way to rationalize binary relations that may fail to admit a weak order consistent with all pairwise comparisons. 
Our method differs from the Schwartz-set approach in two respects: first, it uses the bidirectional pairwise-comparison network rather than only the majority relation, and second, it successively eliminates 
%%F the weakest pairwise dominance relations. 
weaker pairwise dominance relations. 
Eventually, the direct sum of the equivalence classes 
%%F (added)
(survival sets)
generated at each stage of the algorithm coincides with the Schwartz set.

Our equivalence results are also related to the computational literature on graph-theoretic optimization problems. Computing the Schulze winner set is an instance of the all-pairs bottleneck path problem, a classical graph problem in which one seeks, for every pair of vertices, a path maximizing the minimum edge capacity along the path. \citet{vassilevska2009all} established a close connection between this problem and the $(\max,\min)$ matrix product and obtained the first truly subcubic algorithm for dense directed graphs. 
This line of research was further developed by \citet{duan2009fast}. 
More recently, \citet{cornect2024gromov} related the bottleneck path problem on undirected graphs to the computation of Gromov approximating trees for finite metric spaces. 
As discussed in Section~\ref{sec:computational}, the currently best randomized expected-time algorithm for computing the Schulze winner set is due to \citet{SVX2021}, which runs in ${\rm O}(m^2 \log^4 m)$ time.

\section{The Schulze rule}\label{sec: Schulze}
Let $N=\{1,2,\ldots,n\}$ be the finite 
nonempty %%F (added)
set of voters, and let $X=\{x_1,x_2,\ldots,x_m\}$ be a finite nonempty set of alternatives (or candidates).
We assume that each voter $i\in N$ has a preference relation, namely a 
%%F linear ordering 
linear order 
$P_i$ on $X$.
%%F (added below)
That is, $xP_iy$ for $x,y\in X$ means that voter $i$ prefers $x$ to $y$. %%F
Let $\mathcal{P}$ be the set of all 
%%F linear orderings on $X$, 
linear orders on $X$, %%F 
and let $\mathcal{P}^n=\prod_{i\in N}\mathcal{P}$ be the set of all preference profiles.

For each $P_i\in\mathcal{P}$, we define a ranking function $r_i:X\rightarrow\{1,2,\ldots,m\}$
such that
\[
r_i(x)=m-\left|\{y\in X \mid xP_i y\}\right|,
\qquad \forall x\in X.
\]
Since there is a one-to-one correspondence between $P_i \in \mathcal{P}$ and $r_i$, hereafter we identify a ranking function as a preference of a voter.
For every distinct pair of alternatives $(x,y) \in X \times X$, define the function $\gamma:X\times X\rightarrow \mathbb{Z}_{\ge 0}$
by
\[
\gamma(x,y)
=
\left|
\left\{
i\in N \mid r_i(x)<r_i(y)
\right\}
\right|.
\]

The function $\gamma$ is called the \emph{pairwise-record function}, and the number $\gamma(x,y)$ is referred to as the \emph{pairwise record} of $x$ versus $y$, namely the number of voters who prefer $x$ to $y$.
Note that
\begin{equation}\label{eq:v1a}
\gamma(x,y)+\gamma(y,x)=n
\end{equation}
for any $(x,y)\in X\times X$ with $x\neq y$.
An alternative $x\in X$ is called the {\it Condorcet winner (resp. loser)} if for every $y \in X\setminus\{x\}$ we have $\gamma(x,y)>\gamma(y,x)$ (resp. $\gamma(x,y)>\gamma(y,x)$). 
By definition, the Condorcet winner (resp. loser) is unique if it exists.
%Similarly, an alternative $x \in X$ the {\it Condorcet loser} if for every $y \in X\setminus\{x\}$ we have $\gamma(x,y)<\gamma(y,x)$, which is unique if it exists.

Let
\[
A_{\gamma}
=
\{
(x,y)\in X\times X
\mid
\gamma(x,y)>0
\}
\]
be the set of all ordered pairs whose pairwise record is positive.
The pairwise-record function $\gamma$ induces a weighted directed graph, referred to as the \emph{pairwise-record graph}, $G_{\gamma}=(X,A_{\gamma})$, where $X$ is the set of vertices and $A_{\gamma}$ is the set of arcs.\footnote{See the Appendix for graph-theoretic preliminaries.}
We say that the weighted graph $\mathcal{N}=(G_{\gamma}, \gamma)$ is a {\it pairwise comparison network}.
It is remarkable that the {\it majority graph} $G_{\mu}=(X, B_{\mu})$ is a subgraph of the corresponding pairwise record graph $G_{\gamma}$ by removing all arcs $(x,y)\in A_{\gamma}$ with $\gamma(x,y)<\gamma(y,x)$, that is, $B_{\mu}=\{(x,y) \in X \times X \mid \mu(x,y)>0\}$, where $\mu$ is the {\it majority margin} $\mu(x,y)\equiv\gamma(x,y)-\gamma(y,x)$.
We also say that the weighted graph $\mathcal{N}_M=(G_{\mu}, \mu)$ is a {\it majority margin network}.

For any $(x,y)\in X\times X$, define the \emph{strength} of $x$ against $y$ by
\[
\tau^\gamma_{G_{\gamma}}(x,y)
=
\begin{cases}
\max_{\pi \in \Pi(x,y)}\min_{(w,z)\in \pi}\gamma(w,z)
~\text{if}~\Pi(x,y) \neq \emptyset,\\
0~\text{otherwise},
\end{cases}
\]
where $\Pi(x,y)$ denotes the set of all directed paths from $x$ to $y$, and $(x,y)\in \pi$ means that the arc $(x,y)$ lies on the path $\pi$.\footnote{$\pi$ is a path from $x$ to $y$ if there exists $z_1, \ldots z_l \in X$ such that $z_1=x, z_l=y$ and $\pi=\{(z_k,z_{k+1}) \in A_{\gamma} \mid k=1, \ldots, l-1\}$.}
Moreover, $\tau^\mu_{G_\mu}(x,y)$ is defined analogously.
The {\it Schulze winner set} \citep{Schulze1997}, denoted by $F^{Schulze}(\mathcal{N})$, is the set of all candidates $x$ that is not beaten by any other $y \in X \setminus\{x\}$ with respect to $\tau^{\gamma}_{G_{\gamma}}$:
\[
F^{Schulze}(\mathcal{N})=\bigl\{x \in X \mid \tau^{\gamma}_{G_{\gamma}}(x,y) \ge \tau^{\gamma}_{G_{\gamma}}(y,x), \forall y \in X \setminus \{x\} \bigr\}.
\]

Mathematically, the problem to find the Schulze winner set corresponds to the widest path, or all-pairs bottleneck path problem in graph theory and we can apply a version of Floyd--Warshall Algorithm \citep{Floyd1962,Warshall1962} to solve the problem.
While it is not self-evident that the Schulze winner set is non-empty, the well-definedness of the Schulze rule is proven in previous literature \citep{Schulze2003, Schulze2011, Schulze2025}.
Hereafter, the non-emptyness of the Schulze winner set is taken as given.
As it is straightforwardly verified, \citet{Schulze1997} shows that the Schulze winner set always includes the Condorcet winner, if it exists.

Although we have defined the Schulze winner set in terms of the pairwise-comparison network, another way to define the Schulze winner set is to define over the majority margin network $\mathcal{N}_M$ as follows:
\[
F^{Schulze}(\mathcal{N}_M)=\bigl\{x \in X \mid \tau^{\mu}_{G_{\mu}}(x,y) \ge \tau^{\mu}_{G_{\mu}}(y,x), \forall y \in X \setminus \{x\} \bigr\}.
\]
Interestingly, 
%%F the two Schulze winner set 
the two Schulze winner sets  %%F 
can be proven to be equivalent.

\begin{proposition}\label{prop:schulze_mar}
For any pairwise comparison network $\mathcal{N}=(G_{\gamma}, \gamma)$ and the corresponding majority margin network $\mathcal{N}_M=(G_{\mu}, \mu)$, we have
\[
F^{Schulze}(\mathcal{N})=F^{Schulze}(\mathcal{N}_M).
\]
\end{proposition}

\section{A new algorithm:  \textit{SuccessiveDicut}}\label{sec:successiveDicut}
We introduce a new algorithm, \emph{SuccessiveDicut}, and show that this algorithm eventually induces the Schulze winner set.
Let $\mathcal{N}=(G_{\gamma}, \gamma)$ be an input for the algorithm.
The procedure, which is formalized by {\sf Procedure} {\bf SuccessiveDicut}$(\mathcal{N})$, is as follows.
We compute the minimum $\gamma^*>0$ in {\bf Step 2a} and remove the arcs that attain the minimum. 
Then we extract the poset structure within the current graph $H=(U,B)$, pick up maximal strongly connected components, discard the other part of the current graph $H$, and update the current $H=(U,B)$.
%%F (Added)
\footnote{See Appendix~\ref{appendix: graph} for the definition of {\it maximal} strongly connected 
component.}  %%F 
We repeat this process until we get $B=\emptyset$.
Let $F^{SD}(\mathcal{N})$ be the output of the algorithm.

\noindent\\
----------------------------------------------------------------------------------------------------\\
{\sf Procedure} {\bf SuccessiveDicut}$(\mathcal{N})$  \\ 
----------------------------------------------------------------------------------------------------\\
{\bf Input}: The pairwise comparison network $\mathcal{N}=(G_{\gamma}, \gamma)$.\\
{\bf Output}: The winner set $F^{SD}(\mathcal{N})$ of $\mathcal{N}$.
\smallskip\\
{\bf Step 1}: Let $H=(U,B)$ be the maximal strongly connected component 
%%F of $G=(V,A)$;\\
of $G_\gamma=(V,A)$;\\  %%F
{\bf Step 2}: While $B\neq\emptyset$, do the following:\\
\hspace*{1em} {\bf Step 2a}: Compute $\gamma^*=\min\{\gamma(a)\mid a\in B\}$;\\
\hspace*{2em} $B^*\gets \{a\in B\mid \gamma(a)=\gamma^*\}$;\\
\hspace*{2em} Delete arcs in $B^*$ from $H=(U,B)$;\\
\hspace*{2em} Update $B$ and denote the obtained graph by $H=(U,B)$ again;\\ 
\hspace*{1em} {\bf Step 2b}: Find all {\sl maximal} strongly connected 
components of $H$;\\
\hspace*{2em} Put the components as  $H_i=(U_i,B_i)$ $(i\in I)$;\\
\hspace*{2em} Let $H=(U,B)$ be the direct sum of $H_i=(U_i,B_i)$ $(i\in I)$;\footnote{Here, $U=\cup_{i\in I}U_i$ and $B=\cup_{i\in I}B_i$.}\\
 {\bf Step 3}: Return $F^{SD}(\mathcal{N})=U$; \\
----------------------------------------------------------------------------------------------------
\noindent\\

Note that the finally obtained graph $H$ consists of isolated vertices,
which form the winner set $F^{SD}(\mathcal{N})$.
The following example illustrates how the algorithm works.

\begin{example}
Consider an example from \citet{doring2026river}.
%%F (Added below)
The margin graph is not strongly connected. We consider $n=50$ with which
we defined $\gamma$ from the majority margin $\mu$ by 
$\gamma(u,v)=(\mu(u,v)+n)/2$. %%F (as in Remark~\ref{rem:MarginGraph}).
%%F  
Performing {\sf{Procedure}} \textbf{SuccessiveDicut}, we obtain (1) in Figure~\ref{Fig_doring} when $\gamma^*=24$.
%%F Correct Figure 1: The arc (f,b) in (1) and (2) should be reversed as (b,f) %%%%%%%%%%%%%%%
The later behavior of {\sf{Procedure}} \textbf{SuccessiveDicut} is shown in Figure~\ref{Fig_doring}, where it proceeds with (2)$\rightarrow$(3)$\rightarrow$(4)$\rightarrow$(5)$\rightarrow$(6).
%%F (Added below)
The broken lines appearing in (2) and (6) represent directed cuts (dicuts).
%%F 
The winner set $F^{SD}$ in this example is $\{b\}$, which corresponds to the Schulze winner.

%================================================
% Figure 4
%================================================

\tikzset{
  E/.style={
    -{Latex[length=1.8mm,width=1.35mm]},
    line width=.55pt,
    shorten >=10.5pt,
    shorten <=4pt
    },
  N/.style={
    font=\scriptsize,
    fill=white,
    inner sep=.7pt
  },
  Circ/.style={
    font=\scriptsize,
    draw,
    circle,
    fill=white,
    inner sep=.5pt
  },
  V/.style={
    font=\Large,
    fill=white,
    inner sep=2pt
  },
}

\def\CoordsF{%
  \coordinate (f) at (0,1);
  \coordinate (a) at (1.35,2.25);
  \coordinate (d) at (1.35,-.25);
  \coordinate (b) at (5.05,2.25);
  \coordinate (c) at (5.05,-.25);
  \coordinate (e) at (6.85,1);
}

\def\CoordsNoF{%
  \coordinate (a) at (0,2.25);
  \coordinate (d) at (0,-.25);
  \coordinate (b) at (3.95,2.25);
  \coordinate (c) at (3.95,-.25);
  \coordinate (e) at (5.85,1);
}

\def\VerticesF{%
  \node[V] at (f) {$f$};
  \node[V] at (a) {$a$};
  \node[V] at (d) {$d$};
  \node[V] at (b) {$b$};
  \node[V] at (c) {$c$};
  \node[V] at (e) {$e$};
}

\def\VerticesNoF{%
  \node[V] at (a) {$a$};
  \node[V] at (d) {$d$};
  \node[V] at (b) {$b$};
  \node[V] at (c) {$c$};
  \node[V] at (e) {$e$};
}

\def\EAB#1{%
  \draw[E] (a)--node[#1,above,pos=.50,yshift=-2pt] {$33$} (b);
}

\def\EAD#1{%
  \draw[E] (a)--node[#1,right,pos=.50,xshift=-3pt,yshift=17pt] {$36$} (d);
}

\def\EBC#1{%
  \draw[E] (b)--node[#1,right,pos=.46,xshift=-5pt,yshift=13pt] {$32$} (c);
}

\def\EDC#1{%
  \draw[E] (d)--node[#1,below,pos=.52,yshift=2pt] {$35$} (c);
}

\def\EBE#1{%
  \draw[E] (b)--node[#1,above,pos=.42,xshift=1pt,yshift=-2pt] {$38$} (e);
}

\def\ECE#1{%
  \draw[E] (c)--node[#1,below,pos=.56,xshift=-5pt,yshift=-2pt] {$39$} (e);
}

\def\EAE#1{%
  \draw[E] (a)--node[#1,above,pos=.76,xshift=1pt,yshift=-3pt] {$40$} (e);
}

\def\EED#1{%
  \draw[E] (e)--node[#1,above,pos=.50,xshift=12pt,yshift=-1pt] {$37$} (d);
}

\def\ECA#1{%
  \draw[E] (c)--node[#1,below,pos=.47,xshift=7pt,yshift=-3pt] {$34$} (a);
}

\def\EDB#1{%
  \draw[E] (d)--node[#1,above,pos=.60,xshift=2pt,yshift=-1pt] {$31$} (b);
}

\def\CoreAll{%
  \EAB{N}
  \EAD{N}
  \EBC{N}
  \EDC{N}
  \EBE{N}
  \ECE{N}
  \EAE{N}
  \EED{N}
  \ECA{N}
  \EDB{N}
}

\def\EdgesF{%
  \draw[E] (a) to[bend right=9]
    node[N,above,pos=.43,xshift=-2pt,yshift=3pt] {$26$} (f);

  \draw[E] (f) to[bend right=9]
    node[Circ,below,pos=.55,xshift=-1pt,yshift=2pt] {$24$} (a);

  \draw[E] (d)--node[N,below,pos=.48,xshift=-2pt,yshift=-3pt] {$27$} (f);

  \draw[E] (e)--node[N,above,pos=.48,xshift=-35pt,yshift=-2pt] {$30$} (f);

  \draw[E] (c)--node[N,below,pos=.55,xshift=-12pt,yshift=6pt] {$28$} (f);

  \draw[E] (b)--node[N,above,pos=.76,xshift=55pt,yshift=12pt] {$29$} (f);
}

\def\EdgesFNoFA{%
  \draw[E] (a) to[bend right=9]
    node[N,above,pos=.43,xshift=-2pt,yshift=3pt] {$26$} (f);

  \draw[E] (d)--node[N,below,pos=.48,xshift=-2pt,yshift=-3pt] {$27$} (f);

  \draw[E] (e)--node[N,above,pos=.48,xshift=-35pt,yshift=-2pt] {$30$} (f);

  \draw[E] (c)--node[N,below,pos=.55,xshift=-12pt,yshift=6pt] {$28$} (f);

  \draw[E] (b)--node[N,above,pos=.76,xshift=55pt,yshift=12pt] {$29$} (f);
}

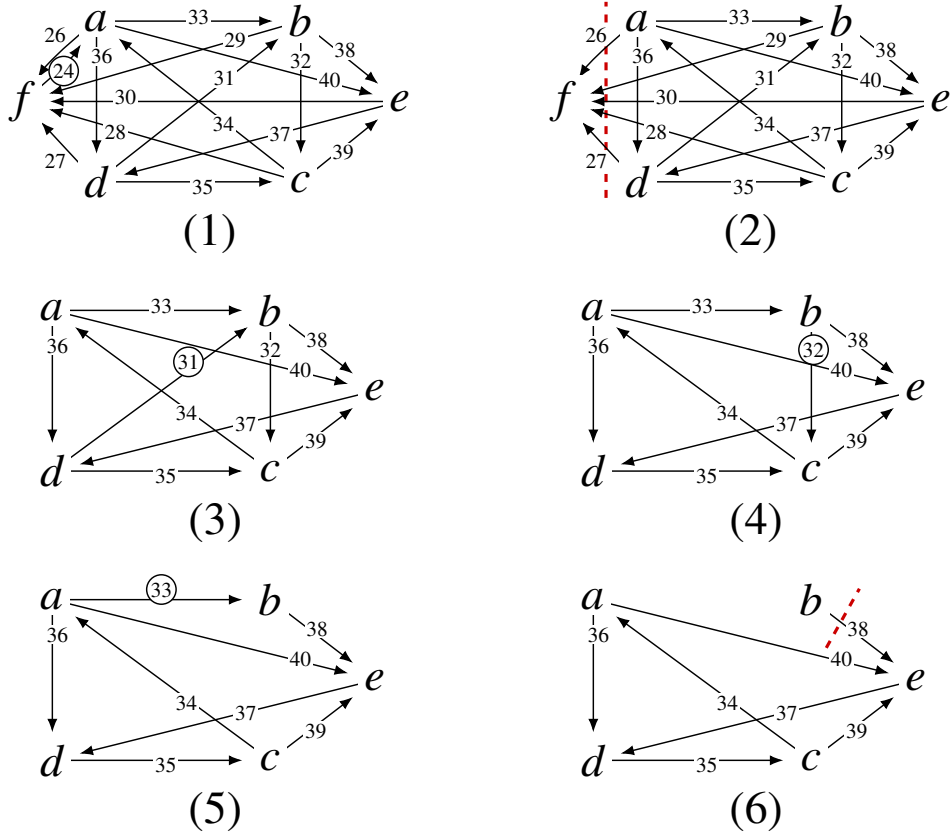
\begin{figure}[!htbp]
\centering

\begin{tikzpicture}[x=.73cm,y=.85cm]

%-----------------------------
% (1)
%-----------------------------
\begin{scope}[shift={(0,9.0)}]
  \CoordsF
  \EdgesF
  \CoreAll
  \VerticesF
  \node[font=\Large] at (3.4,-1.05) {(1)};
\end{scope}

%-----------------------------
% (2)
%-----------------------------
\begin{scope}[shift={(9.8,9.0)}]
  \CoordsF

  \draw[dashed,red!80!black,line width=1.2pt] (.78,-.50)--(.78,2.55);

  \EdgesFNoFA
  \CoreAll
  \VerticesF

  \node[font=\Large] at (3.4,-1.05) {(2)};
\end{scope}

%-----------------------------
% (3)
%-----------------------------
\begin{scope}[shift={(.55,4.5)}]
  \CoordsNoF
  \EAB{N}
  \EAD{N}
  \EBC{N}
  \EDC{N}
  \EBE{N}
  \ECE{N}
  \EAE{N}
  \EED{N}
  \ECA{N}
  \EDB{Circ}
  \VerticesNoF
  \node[font=\Large] at (2.95,-1.05) {(3)};
\end{scope}

%-----------------------------
% (4)
%-----------------------------
\begin{scope}[shift={(10.35,4.5)}]
  \CoordsNoF
  \EAB{N}
  \EAD{N}
  \EBC{Circ}
  \EDC{N}
  \EBE{N}
  \ECE{N}
  \EAE{N}
  \EED{N}
  \ECA{N}
  \VerticesNoF
  \node[font=\Large] at (2.95,-1.05) {(4)};
\end{scope}

%-----------------------------
% (5)
%-----------------------------
\begin{scope}[shift={(.55,0)}]
  \CoordsNoF
  \EAB{Circ}
  \EAD{N}
  \EDC{N}
  \EBE{N}
  \ECE{N}
  \EAE{N}
  \EED{N}
  \ECA{N}
  \VerticesNoF
  \node[font=\Large] at (2.95,-1.05) {(5)};
\end{scope}

%-----------------------------
% (6)
%-----------------------------
\begin{scope}[shift={(10.35,0)}]
  \CoordsNoF
  \EAD{N}
  \EDC{N}
  \EBE{N}
  \ECE{N}
  \EAE{N}
  \EED{N}
  \ECA{N}
  \draw[dashed,red!80!black,line width=1.2pt] (4.22,1.5)--(4.82,2.42);
  \VerticesNoF
  \node[font=\Large] at (2.95,-1.05) {(6)};
\end{scope}

\end{tikzpicture}

\caption{An example from \citet{doring2026river}.}\label{Fig_doring}
\end{figure}
\end{example}

Recall that every margin graph $G_{\mu}=(X,B_{\mu})$ is a subgraph of a corresponding pairwise-record graph $G_{\gamma}=(X,A_{\gamma})$. 
If the margin graph $G_{\mu}$ is strongly 
connected, then during the execution of our {\sf Procedure} 
{\bf SuccessiveDicut}, $G_{\mu}$ appears when value $\gamma^*$ becomes large enough but not larger than $n/2$. 
Indeed, for any $(x,y) \in B_{\mu}$, we have $\mu(x,y)=\gamma(x,y)-\gamma(y,x)>0$ and $\gamma(x,y)\ge n/2 \ge 
\gamma(y,x)$. 
Hence when value $\gamma^*$ becomes 
large enough but not larger than $n/2$, all the arcs in $B_{\mu}$ remains in the 
current graph obtained at an execution of {\bf Step 2b} and the current graph is strongly connected.
Moreover, since $\mu(x,y)=\gamma(x,y)-\gamma(y,x)>0$ and 
$\gamma(x,y)+\gamma(y,x)=n$ by \eqref{eq:v1a}, we have $\mu(x,y)=2\gamma(x,y)-n$.
Hence the ordinal relations for $\mu$ and $\gamma$ are the same 
for the current arc set $B_{\mu}$.  

Note that the operations performed by {\sf Procedure} {\bf SuccessiveDicut}
depend only on the ordinal relation by $\gamma$. 
So, when the margin graph $G_{\mu}=(X,B_{\mu})$ is strongly connected, 
applying our {\sf Procedure} {\bf SuccessiveDicut} to 
the network $\mathcal{N}_M=(G_{\mu},\mu)$ results in the same winner set as that 
obtained by {\sf Procedure} {\bf SuccessiveDicut} to the original  $\mathcal{N}=(G_{\gamma},\gamma)$. 
This observation is summarized in the following proposition, parallel to Proposition~\ref{prop:schulze_mar}.

\begin{proposition}\label{prop: SD_mar}
For any pairwise comparison network $\mathcal{N}=(G_{\gamma}, \gamma)$ and the corresponding majority margin network $\mathcal{N}_M=(G_{\mu}, \mu)$, we have
\[
F^{SD}(\mathcal{N})=F^{SD}(\mathcal{N}_M).
\]
\end{proposition}

\subsection{Equivalence of the Schulze rule}
We are ready to state our main equivalence result.

\begin{theorem}\label{thm:equivalence}
For any pairwise comparison network $\mathcal{N}=(G_{\gamma},\gamma)$, 
\[
F^{SD}(\mathcal{N})=F^{Schulze}(\mathcal{N}).
\]
\end{theorem}

We first show the following lemma for the proof of 
Theorem~\ref{thm:equivalence}.

\begin{lemma}\label{lemma:discarding}
Let $H_1=(U_1,B_1)$ be a maximal strongly connected component of the current 
graph $H=(U,B)$ obtained by an execution of\; {\bf Step 1} or {\bf Step 2b}. 
Also let $H_2=(U_2,B_2)$ be any strongly connected component of the current 
graph $H=(U,B)$ that is smaller than $H_1=(U_1,B_1)$ in the poset structure 
within the current graph $H=(U,B)$.
%%F Every alternative
Then, every alternative %%F
$y \in U_2$ is defeated by any alternative $x \in U_1$ in the sense of the Schulze rule according to $\mathcal{N}_M$.
Hence, the set $U_2$ of weaker candidates can be discarded, i.e., 
%%F $F^{Schulze}(\mathcal{N}_M) \subseteq U_1$.
$F^{Schulze}(\mathcal{N}_M)$ is included in the union of the vertex sets of 
maximal strongly connected components.  %%F
\end{lemma}

\begin{proof}
%%G \noindent {\bf Proof.} 
First, suppose that the initial graph $G_\gamma$ is not strongly connected.
There uniquely exists a maximal strongly connected component 
$H_1=(U_1,B_1)$ due to the definition of $G_\gamma$. 
Let $H_2=(U_2,B_2)$ be any other strongly connected component of $G_\gamma$. 
Then for any $x\in U_1$ and any $y\in U_2$ there exists a path from $x$ to $y$ 
but not from $y$ to $x$. This means that $y$ is defeated by $x$ in the sense 
of the Schulze rule. It follows that the set $U_2$ of weaker candidates can 
be discarded in {\bf Step 1}.

Next, consider an execution of {\bf Step 2b}.
We prove the present lemma in the following three cases. 
Let $\gamma^*$ be the one computed at an execution of {\bf Step 2}. 
\medskip\\
 \underline{Case 1: $\gamma^*<n/2$}.
%%F (added)
In the present case, there uniquely exists a maximal strongly connected 
component $H_1=(U_1,B_1)$ of the current graph $H=(U,B)$, since $H$ is 
connected by of the definition of pairwise comparison graph.
Let $H_2=(U_2,B_2)$ be any other strongly connected component of $H=(U,B)$.
For any $x\in U_1$ and any $y\in U_2$ there exists a path $\pi$ from $x$ to $y$ 
such that the minimum value of $\gamma(u,v)$ over the arcs $(u,v)$ in $\pi$ is 
greater than or equal to $\gamma^*$, but there does not exist such a path 
from $y$ to $x$. This means that $y$ is defeated by $x$ in the sense 
of the Schulze rule. It follows that the set $U_2$ of weaker candidates can 
be discarded.
%%F 
\medskip\\
%%F \noindent \underline{Case 2: $\gamma^*=n/2$}.
 \underline{Case 2: $\gamma^*=n/2$} (only if $n$ is even). %%F
Denote by $H^*=(U^*,B^*)$ the graph $H=(U,B)$ immediately before updating 
$H$ at {\bf Step 2a} and let $H$ be the one updated at {\bf Step 2a}.
Then the difference between $H^*$ and $H$, is that for every arc $(u,v)\in B^*$
satisfying $\gamma(u,v)=n/2$ both $(u,v)$ and $(v,u)$ disappear in $H$ 
and other arcs remain in $H$.
\Delete{
Moreover, the updated $H$ is exactly equal to $G_M$.
If a path $\pi$ from $u$ to $v$ attains $\tau_{H^*}(u,v)$ goes through an
arc $(u,v)\in B^*$ satisfying $\gamma(u,v)=n/2$, then 
$\tau_{H^*}^\gamma(u,v)=n/2$. This means that $\tau_{G_M}^\mu(u,v)=0$.
}
Hence we can treat each connected component of updated $H$ separately
and the proof in Case 1 can be adapted to the present case.
\medskip\\
 \underline{Case 3: $\gamma^*>n/2$}.
In this case, since the current $\gamma^*$ satisfies $\gamma^* > n/2$, 
we have $\mu(x,y) = 2\gamma(x,y) - n$ for all arcs $(x,y)\in B$ in $H = (U,B)$.
Moreover, the collection of maximal strongly connected components 
$H_i = (U_i, B_i)$ $(i \in I)$ of $H$ obtained at \textbf{Step 2b} is the same 
as that obtained for $G_{\mu}$ with all the arcs $(x,y)$ satisfying 
$\gamma(x,y) \le \gamma^*$ deleted and with the updated vertex set $U$.
Hence, similarly as in Case 1 and Case 2, we can prove the lemma in 
the present case.
\medskip\\
\indent Combining the three cases, we complete the proof.
\end{proof}
%%G \pend
\medskip

%%F \begin{proof}
\begin{proof}(\textit{Proof of Theorem \ref{thm:equivalence}})
By the execution of {\bf Step 1}, the obtained $U$ includes all the Schulze winners due to 
%%F Lemma~\ref{lemma:decomposition} and 
Lemma~\ref{lemma:discarding}.
During the execution of the while loop of {\bf Step 2}, we repeatedly remove weaker candidates and keep the current vertex set $U$ that includes all the Schulze winners. 
Here $U$ is the union of the maximal strongly connected components of a current $H$, where $H$ is possibly decomposed into at least two connected components.
The while loop of {\bf Step 2} is repeated at most $m$ times and the finally obtained graph $H=(U,B)$ consists of isolated vertices. 

%%F For every distinct $x,y\in U$, 
For the finally obtained $U$, consider any distinct $x,y\in U$.  %%F
%%F if arc $(x,y)$ 
If arc $(x,y)$ 
disappears because $\tau_{G_\gamma}^\gamma(x,y)=n/2$, 
it means $\tau_{G_{\mu}}^\mu(x,y)=0$. 
Otherwise, if arc $(x,y)$ disappears when $\gamma(x,y)=\gamma^*>n/2$, since both $x$ and $y$ belong to the final $U$, there exist a vertex set $W$ such that $\Delta^+(W)\neq\emptyset\neq\Delta^-(W)$, $(x,y)\in \Delta^+(W)\cup \Delta^-(W)$, and $\gamma(w,z)=\gamma^*$ $(\forall (w,z)\in  \Delta^+(W)\cup \Delta^-(W))$. This means that $\tau_{G_{\gamma}}^\gamma(x,y)=\tau_{G_{\gamma}}^\gamma(y,x)=\gamma^*$, i.e., $\tau_{G_{\mu}}^\mu(x,y)=\tau_{G_{\mu}}^\mu(y,x)$. Therefore, for every distinct $x,y\in U$, $x$ never defeats $y$, so that the final $U$ satisfies $U=F^{Schulze}(\mathcal{N}_M)$. 
By Proposition \ref{prop:schulze_mar}, we obtain $F^{SD}(\mathcal{N})=U=F^{Schulze}(\mathcal{N}_M)=F^{Schulze}(\mathcal{N})$, which completes the proof.
%%F\end{proof}
%%G \pend
\end{proof}

%%%%%%%%%%%%%%%%%%%%%%%%%%%%%%%%%%%%%%%%%%%%%%
\subsection{Equivalence of the Schwartz set}
%%%%%%%%%%%%%%%%%%%%%%%%%%%%%%%%%%%%%%%%%%%%%%

For a graph $G=(V,A)$,
the well-known set-valued outcome called \textit{Schwartz set} for $G$ 
is formally defined as
\[
\text{Schw}(G) \equiv
\bigl\{x \in V \mid
    \forall y \in V:\;
    x \in R(y) \Longrightarrow y \in R(x)
    \bigr\},
\]
where $R(x)$ denotes the set of vertices that can be reached by paths 
from $x$ in $G$. It follows from this definition that 
\begin{itemize}
\item[{\bf (Schw)}] the Schwartz set $\text{Schw}(G)$ is exactly 
the union of the vertex sets of all {\it maximal} strongly connected 
components of $G$.
\end{itemize}
For, if $x$ is a vertex of a maximal strongly connected component, then there 
holds: $\forall y \in V:\; x \in R(y) \Longrightarrow y \in R(x)$. On the 
other hand, if $x$ is a vertex of a strongly connected component that 
is not maximal, then there exists a vertex $y$ in some maximal strongly 
connected component such that we have $x \in R(y)$ but not $y \in R(x)$. 
(See Appendix \ref{appendix: graph} for the definition of {\it maximal} strongly connected 
component.)

In our algorithm, we %%F iteratively 
execute \textbf{Step~2} until the set $B$ becomes empty. Let 
%%F $T$ 
$k$ %%F
be the total number of executed while loops. 
%%F performed.
Then, %%F we denote 
let
%%F $H^{t}_k = (U^t_k, B^t_k)$ with $k \in \{a,b\}$ 
$H^{(\ell)}_a = (U^{(\ell)}_a, B^{(\ell)}_a)$ and 
$H^{(\ell)}_b = (U^{(\ell)}_b, B^{(\ell)}_b)$, respectively, 
%%F with $k \in \{a,b\}$ 
%%F \textbf{Step~2k} in iteration for each round $t\in\{1,2,\dots,T\}$.
%%F be the set %%F
be the current graphs $H=(U,B)$  %%F
obtained at the end of the $\ell$th execution of \textbf{Step~2a} and
\textbf{Step~2b}.
We also define $U^{(0)}_b$ to be the set $U$ obtained by {\bf Step 1}.
By the above discussion applying to $\mathcal{N}=(G_{\gamma},\gamma)$ and Theorem \ref{thm:equivalence}, we obtain the following characterization of the Schulze winner set in terms of successive formations of the Schwartz set.

\begin{corollary}
For any pairwise comparison network $\mathcal{N}=(G_{\gamma},\gamma)$, the following three statements hold.
\begin{enumerate}
  \item[{\rm (i)}] %%F $U^1_a = \text{Schw}(G_\gamma)$.
$U^{(0)}_b = {\rm Schw}(G_\gamma)$.
  \item[{\rm (ii)}] %%F In each iteration $t$ $(1 \le t \le T)$ of \textbf{Step~2},
In each $\ell$th execution of the while loop of\;
%%F \textbf{Step~2},
 {\bf Step~2},
%%F  we have $\text{Schw}(H^{t}_a) = U^t_b$.
 we have ${\rm Schw}(H^{(\ell)}_a) = U^{(\ell)}_b$ for $\ell=1,2,\cdots,k$.
  \item[{\rm (iii)}] %%F The sequence $\{ U^t_b \}_{t=1}^T$ 
The sequence $\{ U^{(\ell)}_b \}_{\ell=0}^k$ %%F
is monotonically non-increasing with respect to set inclusion, 
%%F and $U^T_b = F^{Schulze}(\mathcal{N})$.
and $U^{(k)}_b = F^{Schulze}(\mathcal{N})$.
\end{enumerate}
\end{corollary}

%%F 
\begin{proof}
%\noindent {\bf Proof.}
By {\bf Step 1} of our algorithm, we can obtain a linearly ordered set of strongly connected components $H_i$ $(i\in I)$ of $G_{\gamma}$ (see Proposition \ref{prop: poset}). Hence, (i) holds.
Moreover, \textbf{(Schw)} and Theorem~\ref{thm:equivalence} directly 
%%F implies 
imply 
(ii) and (iii).
%%F 
\end{proof}
%\pend\medskip %%F

%%G By {\bf Step 1} of our algorithm, we can obtain a linearly ordered set of strongly connected components $H_i$ $(i\in I)$ of $G_{\gamma}$ (see Proposition \ref{prop: poset}). 
%%G Hence the set $U$ obtained by {\bf Step 1} of {\sf Procedure} 
%%G {\bf SuccessiveDicut} is the Schwartz set $\text{Schw}(G_\gamma)$.
%%G Also, the set $U$ obtained at the end of an execution of {\bf Step 2b} is the Schwartz set $\text{Schw}(H)$ with $H$ being the current graph $H$ updated at the end of {\bf Step 2a}.

%%F This implies 
We now see  %%F 
that Schulze rule is an iterative procedure successively forming the 
%%F schwartz set.
Schwartz set. %%F
The notion of Schwartz set is an interpretation by 
%%F \citet{Schwartz1972} 
%%G Schwartz (1972) %%F 
\citet{Schwartz1972}
%%G Schwartz (1972) %%F
of the maximazation-based rationalization, widely used in economics, to environments with cyclical preferences.
Although this alternative interpretation of Schulze rule 
%%F is already acknowledged by 
seems to have been already acknowledged by 
%%F \citet{Schulze2003}, 
%%G Schulze (2003), %%F
\citet{Schulze2003},
%%G Schulze (2003), %%F
the proof is not given.
%%F We establish 
We have thus established this fact here. 
%%F , partly to provide a proof for a statement treated only informally in the past literature.

\section{Successive elimination, Kemeny--Young, and Condorcet}\label{sec: KM}

As already mentioned, computing the Schulze winner set can be formulated as an instance of the all-pairs bottleneck path problem. 
In contrast, our procedure {\bf SuccessiveDicut} admits an interpretation as a successive cycle-elimination method on the pairwise majority graph.
This perspective is closely related to the Kemeny--Young method, which resolves cyclic majority relations by constructing a globally consistent linear ordering.

To clarify this relationship, consider a directed cycle $C$ in the pairwise-comparison network $\mathcal{N}=(G_{\gamma},\gamma)$.
Removing one arc from the cycle breaks the contradiction and leaves a directed path.
The quality of such a deletion can be evaluated using the notion of {\sl support} introduced by \citet{Young1995}.
Given a cycle $C$ in $G_{\gamma}$ and an arc $(\hat{x},\hat{y})\in C$, deleting $(\hat{x},\hat{y})$ produces a directed path $C(\hat{x},\hat{y})$ from $\hat{y}$ to $\hat{x}$.
The support of this path is defined by
\begin{equation}\label{eq:rem1}
 \sigma(C(\hat{x},\hat{y}))
 \equiv
 \sum_{(x,y)\in C(\hat{x},\hat{y})}\gamma(x,y)
 - \gamma(\hat{x},\hat{y})+n,
\end{equation}
which can be rewritten as
\begin{equation}\label{eq:rem2}
 \sigma(C(\hat{x},\hat{y}))
 =
 \sum_{(x,y)\in C}\gamma(x,y)
 -2\gamma(\hat{x},\hat{y})+n.
\end{equation}

Suppose that
\[
(x^*,y^*)
\in
\arg\min_{(x,y)\in C}\gamma(x,y),
\]
that is, $(x^*,y^*)$ is the weakest arc in the cycle $C$.
Equation~(\ref{eq:rem2}) implies that removing $(x^*,y^*)$ maximizes the support of the remaining path.
Thus, {\bf SuccessiveDicut} may be interpreted as iteratively resolving cyclic contradictions by preserving the strongest remaining support structure.
In this sense, the procedure provides a complementary perspective to the all-pairs bottleneck path approach underlying the Schulze rule.

This graph-theoretic interpretation also suggests a different reading of Condorcet's original proposal from that of \citet{Young1988}. 
Given Condorcet’s overarching epistemic aim of maximizing the probability of truth-conformity, Young gives a maximum-likelihood interpretation of Condorcet's method and connects it to the problem of recovering the most plausible complete ranking.
By contrast, our interpretation focuses on Condorcet's discussion of resolving inconsistencies among majority-supported propositions.
If the aim is to choose the best set of alternatives, excluding all global inconsistencies may be an unnecessarily strong requirement.
Therefore, Condorcet distinguishes the consistency required for determining a complete social ranking from the weaker consistency sufficient for selecting a single winner.\footnote{
Condorcet writes:
\begin{quote}
There are two questions to be distinguished here. One may ask whether there is any contradiction in the whole system of propositions; but one may also ask whether there is any contradiction only in that part of the system which is sufficient to decide the superiority of one candidate over all the others.
\citep[pp.~clxx--clxxi]{Condorcet1785}
\end{quote}
}
From this perspective, the objective of the aggregation procedure is not necessarily to construct a complete linear order, but rather to retain an acyclic subsystem of majority-supported propositions sufficient to identify a maximal alternative.

%%F Our {\bf SuccessiveDicut} procedure 
Our Procedure {\bf SuccessiveDicut} %%F 
formalizes this interpretation by successively eliminating weak majority relations until a stable acyclic structure remains.
The equivalence with the Schulze winner set therefore provides a new Condorcetian interpretation of the Schulze rule.
In this sense, the Schulze rule may be viewed as a modern formalization of Condorcet's idea of resolving cyclic majority contradictions through successive elimination of weak propositions.

The distinction between our method and the Kemeny--Young rule reflects the difference between two distinct classes of graph-theoretic optimization problems: bottleneck path problems and linear ordering problems.
Whereas {\bf SuccessiveDicut} operates through iterative elimination of weak pairwise relations, the Kemeny--Young method solves a global ranking optimization problem.

Finding a social ordering according to the Kemeny--Young method is known to be equivalent to minimizing the Kendall tau distance.
For any two linear orders $P,P'\in\mathcal{P}$, the Kendall-tau distance is defined by
\[
\mathcal{K}(P,P')
=
\left|
\left\{
(x,y)\in X\times X
\mid
xPy
\text{ and }
yP'x
\right\}
\right|.
\]
The Kemeny--Young method determines the social preference relation by minimizing the total Kendall-tau distance to the voters' preferences in the following sense.
\[
P^*
\in
\arg\min_{P\in\mathcal{P}}
\sum_{i\in N}\mathcal{K}(P,P_i).
\]
The following proposition shows that this problem can be formulated as the linear ordering problem on the pairwise-comparison network.

\begin{proposition}\label{prop:KY}
Given a profile $(P_1,\ldots,P_n)\in\mathcal{P}^n$ and its pairwise-comparison network $\mathcal{N}=(G_{\gamma},\gamma)$,
\[
\arg\min_{P\in\mathcal{P}}
\sum_{i\in N}\mathcal{K}(P,P_i)
=
\arg\max_{P\in\mathcal{P}}
\sum_{(x,y)\in P}\gamma(x,y).
\]
\end{proposition}

\begin{proof}
By definition,
\begin{align*}
\sum_{i\in N}\mathcal{K}(P,P_i)
&=
\sum_{(x,y)\in P}
|\{i\in N\mid yP_i x\}| \\
&=
\sum_{(x,y)\in P}\gamma(y,x) \\
&=
\sum_{(x,y)\in P}(n-\gamma(x,y)),
\end{align*}
where the last equality follows from \eqref{eq:v1a}.
Therefore, minimizing $\sum_{i \in N} \mathcal{K}(P, P_i)$ is equivalent to maximizing the last term, that is,
\[
\arg\min_{P \in \mathcal{P}} \sum_{i \in N} \mathcal{K}(P, P_i) = \arg\max_{P \in \mathcal{P}} \sum_{(x,y) \in P} \gamma(x,y).
\]
\end{proof}

In this sense, the Kemeny--Young rule and the Schulze rule represent two contrasting approaches to resolving cyclic majority relations: the former seeks a globally optimal linearization of pairwise comparisons, whereas the latter preserves pairwise dominance information through strongest-path propagation and successive elimination of weak contradictions.

This distinction also appears in the structure of the resulting social ranking.
The Kemeny--Young method always produces a complete linear order, whereas {\sf{Procedure}} {\bf SuccessiveDicut} naturally induces a hierarchical decomposition of alternatives through recursive elimination of cyclic majority relations.
More precisely, each stage of the procedure produces a linearly ordered family of strongly connected components, and recursively applying the procedure to each component yields a rooted tree structure of aggregated rankings.
The ranking-theoretic implications of this recursive decomposition are discussed further in Subsection~\ref{sec:aggregated_ranking}.

In generic situations where all pairwise supports are distinct, this hierarchical structure degenerates into a path and therefore determines a complete linear ordering.
Such situations are likely to arise when the number of voters is sufficiently larger than the number of alternatives, since ties among pairwise supports then become increasingly unlikely.

\section{Discussion}\label{sec:discussion}

\subsection{The aggregated ranking}\label{sec:aggregated_ranking}
Recall that, to induce the maximal strongly connected components, we can obtain a linearly ordered family of strongly connected components  $H_1 \succ H_2 \succ \cdots \succ H_k$
of the pairwise-comparison graph $G_{\gamma}$ by {\bf Step~1} of {\sf{Procedure}} {\bf SuccessiveDicut} .
If some component $H_i$ contains more than one alternative, we may recursively apply {\bf SuccessiveDicut} to the subgraph induced by $H_i$.
In this way, the procedure generates a hierarchical decomposition of the candidate set, which naturally induces a rooted tree structure.
See Figure~\ref{fig_continued}.

This recursive decomposition may be interpreted as an aggregated ranking structure.
Even when the majority relation fails to determine a complete linear order globally, the procedure still identifies successive layers of dominance among alternatives.
Thus, the output should be understood not merely as a winner set, but as a hierarchical ranking obtained through iterative resolution of cyclic majority relations.
In particular, if all values of $\gamma$ are distinct, then every elimination step uniquely determines the next decomposition.
Consequently, the induced tree structure degenerates into a path, and the procedure yields a complete linear ordering of alternatives.
Such a situation is likely to occur when the number of voters is sufficiently larger than the number of alternatives, for example when $m \ll n$, since ties among pairwise supports then become increasingly unlikely.

%================================================
% Figure 10 continued
%================================================
\tikzset{
  E/.style={
    -{Latex[length=1.8mm,width=1.35mm]},
    line width=.55pt,
    shorten >=10.5pt,
    shorten <=4pt
  },
  N/.style={
    font=\scriptsize,
    fill=white,
    inner sep=.7pt
  },
  Circ/.style={
    font=\scriptsize,
    draw,
    circle,
    fill=white,
    inner sep=.5pt
  },
  V/.style={
    font=\Large,
    fill=white,
    inner sep=2pt
  },
  Cut/.style={
    dashed,
    red!80!black,
    line width=1.2pt
  },
}

%-----------------------------
% Coordinates
%-----------------------------
\def\CoordsACDE{%
  \coordinate (a) at (0,2.45);
  \coordinate (d) at (0,0);
  \coordinate (c) at (3.75,0);
  \coordinate (e) at (5.95,1.18);
}

\def\CoordsDCE{%
  \coordinate (d) at (0,0);
  \coordinate (c) at (3.55,0);
  \coordinate (e) at (5.65,1.18);
}

\def\CoordsDCEB{%
  \coordinate (d) at (0,0);
  \coordinate (c) at (4.05,.02);
  \coordinate (e) at (5.90,1.30);
}

\def\CoordsDE{%
  \coordinate (d) at (0,0);
  \coordinate (e) at (5.15,.92);
}

%-----------------------------
% Vertices
%-----------------------------
\def\VerticesACDE{%
  \node[V] at (a) {$a$};
  \node[V] at (d) {$d$};
  \node[V] at (c) {$c$};
  \node[V] at (e) {$e$};
}

\def\VerticesDCE{%
  \node[V] at (d) {$d$};
  \node[V] at (c) {$c$};
  \node[V] at (e) {$e$};
}

\def\VerticesDE{%
  \node[V] at (d) {$d$};
  \node[V] at (e) {$e$};
}

%-----------------------------
% Edges
%-----------------------------
\def\EAD#1{%
  \draw[E] (a)--node[#1,right,pos=.45,xshift=-3pt,yshift=0pt] {$36$} (d);
}

\def\EDC#1{%
  \draw[E] (d)--node[#1,below,pos=.52,yshift=-2pt] {$35$} (c);
}

\def\EED#1{%
  \draw[E] (e)--node[#1,above,pos=.50,xshift=8pt,yshift=-1pt] {$37$} (d);
}

\def\ECE#1{%
  \draw[E] (c)--node[#1,below,pos=.57,xshift=-5pt,yshift=-1pt] {$39$} (e);
}

\def\EAE#1{%
  \draw[E] (a)--node[#1,above,pos=.77,xshift=1pt,yshift=-3pt] {$40$} (e);
}

\def\ECA#1{%
  \draw[E] (c)--node[#1,above,pos=.52,xshift=-6pt,yshift=-1pt] {$34$} (a);
}

\def\EdgesACDEAll{%
  \EAD{N}
  \EDC{N}
  \EED{N}
  \ECE{N}
  \EAE{N}
  \ECA{Circ}
}

\def\EdgesACDENoCA{%
  \EAD{N}
  \EDC{N}
  \EED{N}
  \ECE{N}
  \EAE{N}
}

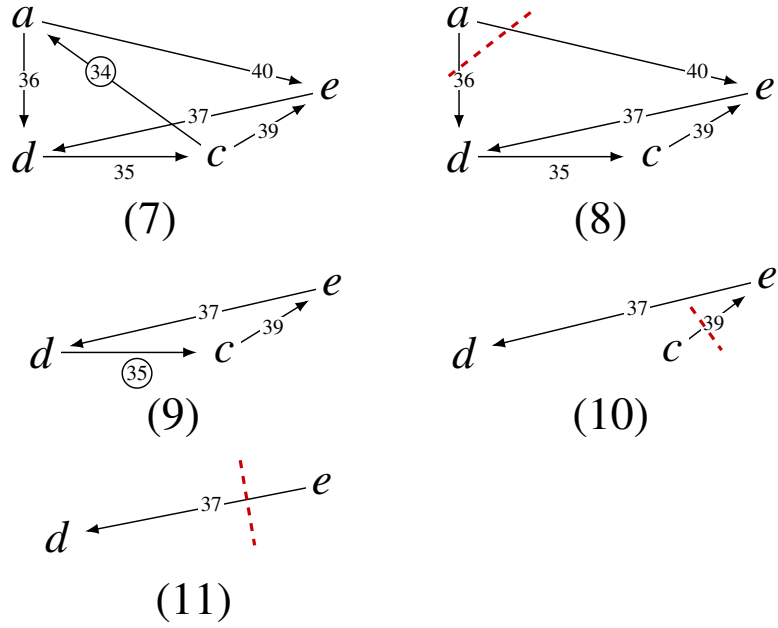
\begin{figure}[!htbp]
\centering

\begin{tikzpicture}[x=.68cm,y=.75cm]

%-----------------------------
% (7)
%-----------------------------
\begin{scope}[shift={(0,6.7)}]
  \CoordsACDE
  \EdgesACDEAll
  \VerticesACDE
  \node[font=\Large] at (2.45,-1.15) {(7)};
\end{scope}

%-----------------------------
% (8)
%-----------------------------
\begin{scope}[shift={(8.45,6.7)}]
  \CoordsACDE
  \EdgesACDENoCA

  \draw[Cut] (-.20,1.42)--(1.50,2.62);

  \VerticesACDE
  \node[font=\Large] at (2.75,-1.15) {(8)};
\end{scope}

%-----------------------------
% (9)
%-----------------------------
\begin{scope}[shift={(.35,3.25)}]
  \CoordsDCE
  \EDC{Circ}
  \EED{N}
  \ECE{N}
  \VerticesDCE
  \node[font=\Large] at (2.55,-1.15) {(9)};
\end{scope}

%-----------------------------
% (10)
%-----------------------------
\begin{scope}[shift={(8.55,3.25)}]
  \CoordsDCEB
  \EED{N}
  \ECE{N}

  \draw[Cut] (4.4,0.8)--(5,.04);

  \VerticesDCE
  \node[font=\Large] at (2.85,-1.15) {(10)};
\end{scope}

%-----------------------------
% (11)
%-----------------------------
\begin{scope}[shift={(.65,0)}]
  \CoordsDE
  \EED{N}

  \draw[Cut] (3.55,1.35)--(3.83,-.15);

  \VerticesDE
  \node[font=\Large] at (2.65,-1.15) {(11)};
\end{scope}

\end{tikzpicture}

\caption{An example from \citet{doring2026river} (continued). The remaining steps of algorithm %%%F procedes 
proceedes with (7) $\rightarrow$ (8) $\rightarrow$ (9) $\rightarrow$ (10) $\rightarrow$ (11).
Then, the final social ordering $P$ corresponds to $b P a P c P e P d$.}
\label{fig_continued}
\end{figure}

\subsection{Satisfaction of the new axioms}
\cite{Schulze2003,Schulze2011,Schulze2025} extensively investigate the axiomatic properties of the Schulze rule.
Thanks to the equivalence between our new algorithm and the Schulze winner set, we can show that the Schulze rule also satisfies the following axiom, called \textit{Schwartz-IIA}.\footnote{The name of this axiom is inspired by the \text{Smith-IIA} criterion introduced in \citet{Schulze2003}. A smith set is defined as $\text{Smith}(G) \equiv \{ x \in V \mid \forall y \in V \colon y \in R(x)\}$ over any graph $G = (V,A)$.} To the best of our knowledge, it has not previously been formally shown that the Schulze rule satisfies this condition.

\iffalse
Consider an alternative $x$ that is Pareto-dominated by every other alternative. That is,
\[
\forall y \in A \setminus \{x\},\quad y \succ_i x
\quad \text{for all } i \in N.
\]
We call the alternative \textit{unanimously Pareto-dominated alternative}, which can be strategically easy to introduce into a given election.
A desirable property of a voting rule is that its outcome be independent of the existence of such an alternative, which is formally defined as follows.
\begin{itemize}
\item[] \textbf{Independence of Unanimously Pareto-dominated Alternative:} If $y \succ_i x$ for all $y \in A \backslash \{x\}$ and $i \in N$, then $F(\mathcal{N}) = F(\mathcal{N}|_{X\setminus\{x\}})$, where $\mathcal{N}|_{X\setminus\{x\}}=({X \backslash \{x\}}, A_{\gamma}|_{X \backslash \{x\}},\gamma|_{X\backslash \{x\}})$.
\end{itemize}
\fi

\begin{itemize}
\item[] \textbf{Schwartz-IIA:} If $x \notin \text{Schw}(G_{\gamma})$,\\ then $F(\mathcal{N}) = F(\mathcal{N}|_{X\setminus\{x\}})$, where $\mathcal{N}|_{X\setminus\{x\}}=({X \backslash \{x\}}, A_{\gamma}|_{X \backslash \{x\}},\gamma|_{X\backslash \{x\}})$.
\end{itemize}

\begin{proposition}\label{prop: IIA}
The Schulze rule satisfies \textit{Schwartz-IIA}.
\end{proposition}

\begin{proof}
Suppose that $x$ is not part of the maximal strongly connected component of the corresponding digraph. Then, it is eliminated in \textbf{Step~1} of the procedure, and is never chosen.
Moreover, the existence of such an alternative does not alter the result of \textbf{Step 1} of the procedure, so that Schulze rule returns the same winner set regardless of its existence.
\end{proof}

Thanks to Proposition \ref{prop: IIA}, the Schulze rule can also be shown to satisfy the following two desirable properties, which have not been investigated in the literature.

First, we say that a winner set satisfies independence of Condorcet losers if it is invariant under the removal of a Condorcet loser.

\begin{itemize}
\item[] \textbf{Independence of Condorcet Losers}: If $|\{i \in N \mid y \succ_i x\}| > |\{i \in N \mid x \succ_i y \}|$ for all $y \in X \backslash \{x\}$, then $F(\mathcal{N}) = F(\mathcal{N}|_{X \backslash \{x\}})$, where $\mathcal{N}|_{X \backslash \{x\}}=({X \backslash \{x\}}, A_{\gamma}|_{X \backslash \{x\}},\gamma|_{X\backslash \{x\}})$.
\end{itemize}
Since a Condorcet loser does not belong to the Schwartz set of $G_{\gamma}$, \textit{Schwartz-IIA} implies that the Schulze rule satisfies this property.

Next, consider an alternative $x$ that is Pareto-dominated by every other alternative. We call such an alternative a \textit{unanimously Pareto-dominated alternative}. Such an alternative may be strategically easy to introduce into a given election. A desirable property of a voting rule is that its outcome be independent of the existence of such an alternative. This property is defined as follows.

\begin{itemize}
\item[] \textbf{Independence of Unanimously Pareto-Dominated Alternatives}: If $y \succ_i x$ for all $y \in X \backslash \{x\}$ and $i \in N$, then $F(\mathcal{N}) = F(\mathcal{N}|_{X\setminus\{x\}})$, where $\mathcal{N}|_{X \backslash \{x\}}=({X \backslash \{x\}}, A_{\gamma}|_{X \backslash \{x\}},\gamma|_{X\backslash \{x\}})$.
\end{itemize}
By the same logic as for \textit{Independence of Condorcet Losers}, since a unanimously Pareto-dominated alternative does not belong to the Schwartz set of $G_{\gamma}$, \textit{Schwartz-IIA} implies that the Schulze rule satisfies this property.

This axiom should be compared with \textit{Independence of Pareto-Dominated Alternatives}, which has been used in various strands of the previous literature.\footnote{See, for example, \citet{Richelson1978}, \citet{Ching1996}, and \citet{Ozturk2020}. It is also comparable to \textit{Independence of Never-Approved Alternatives}, studied by \citet{BrandlPeters2022}, and \textit{Independence of Unanimous Losers}, studied by \citet{Kondratev2023}, although the definitions are slightly different. The corresponding axiom in \citet{BrandlPeters2022} is an independence condition concerning alternatives that receive no approval votes in approval voting, whereas \citet{Kondratev2023} defines the axiom for social preference functions.}
Whereas \textit{Independence of Pareto-Dominated Alternatives} requires that the existence of any alternative Pareto-dominated by some other alternative not alter the winner set, our axiom requires invariance only with respect to alternatives Pareto-dominated by all other alternatives. Note that \cite{doring2026river} shows that the Schulze rule violates the stronger version.

\subsection{Computational issues}\label{sec:computational}
Research on algorithms for finding the Schulze winner set has been intensively made until recently 
%%F (see \citealt{Schulze2025}). 
(see \citealt{Schulze2025}).  %%F
The Floyd-Warshall algorithm for all-pairs bottleneck problem has been  improved and now runs in ${\rm O}(m^{2.69})$ time to find  the Schulze winner set. 
Currently the best, but randomized, expected ${\rm O}(m^2\log^4m)$ time algorithm is also given in 
%%F \citet{SVX2021}.
%%G Sornat et al.~(2021).  %%F
\citet{SVX2021}.

Our \textbf{SuccessiveDicut} method described in Section~\ref{sec:successiveDicut} is a prototype.\footnote{Using the 
sophisticated technique given in \cite{BernsteinGS2021} for updating decremental strongly-connected components, we can perform {\sf Procedure} {\bf SuccessiveDicut} in ${\rm O}(m^{2+\tfrac{2}{3}+{\rm o}(1)})$ time.} 
We can implement the {\bf SuccessiveDicut} method in ${\rm O}(m^3\log m)$ time in a primitive way as follows.
First, we sort the values of $\gamma(u,v)$ $((u,v)\in A)$ and 
let the distinct values of $\gamma(u,v)$ $((u,v)\in A)$ be given by
\begin{equation}\label{eq:comp1}
 \gamma^{(1)}< \gamma^{(2)}< \cdots < \gamma^{(k)}.
\end{equation}
We try to find an integer $\ell$ with $1\le \ell \le k$ such that
\begin{itemize}
\item removing all arcs $(u,v)$ satisfying $\gamma(u,v)< \gamma^{(\ell)}$
from the graph $G$, the graph remains strongly connected, while removing all arcs $(u,v)$ satisfying $\gamma(u,v)\le \gamma^{(\ell)}$ from the graph $G$, the graph becomes not strongly connected.
\end{itemize}
We can find such an integer $\ell$, if any exists, by a binary search in  list (\ref{eq:comp1}) in ${\rm O}(m^2\log m)$ time, where ${\rm O}(m^2)$ is for the strong connectivity checking and ${\rm O}(\log m)$ is for the binary search. Repeating this procedure for updated graph $H$, we can perform 
the {\bf SuccessiveDicut} method in ${\rm O}(m^3\log m)$ time, since the number of such integers $\ell$s is less than $m$ because of repeated deletion of weaker candidates in {\bf Step 2b}.

Moreover, we can start our prototype procedure \textbf{SuccessiveDicut}
by adding preprocessing {\bf Step 0} and modifying {\bf Step 1} as follows.
(Here we assume that there exists $a\in A$ with $\gamma(a)< n/2$, 
excluding the special case when $\gamma(a)= n/2$ for all $a\in A$.)
\medskip\\
%%%%%%%%%%%%%%%%%%%%%%%% New Step 1 and Modified Step 1 %%%%%%%%%%%%%%%%%%%%
{\bf Step 0}: Compute 
    $\gamma*=\max\{\gamma(a)\mid a\in A,\; \gamma(a)< n/2\}$;\\
\hspace*{1em} $B^*\gets \{a \in A\mid \gamma(a)\le \gamma^*\}$;\\
\hspace*{1em} Let $G^*=(V,A\setminus B^*)$;\\
{\bf Step 1}: Let $H=(U,B)$ be the maximal strongly connected component 
of $G^*$;
%%%%%%%%%%%%%%%%%%%%%%%% New Step 1 and Modified Step 1 %%%%%%%%%%%%%%%%%%%%
\medskip\\
The present modification is valid, since every dicut of a current $H$ computed
by an execution of {\bf Step~2a} of the prototype 
procedure \textbf{SuccessiveDicut} while $\gamma^*$ computed at {\bf Step~2a}
satisfies $\gamma^*< n/2$, is also a dicut of $G^*$ obtained by 
{\bf Step~0} given above.

\section*{Acknowledgements} %% (Added below)

S.~Fujishige's research was supported by JSPS KAKENHI Grant Number 
JP22K11922 and by JST ERATO Grant Number JPMJER2301, 
and also by the Research Institute for Mathematical Sciences, an International Joint Usage/Research Center located in Kyoto University.
Satoshi Nakada's research was supported by JSPS KAKENHI Grant Number 25K16606 and 25K00618.

%%F 

%%%%%%%%%%%%%
%%%%%%%%%%%%%APPENDIX
%%%%%%%%%%%%%

\begin{center}
\Large{{\bf Appendix}}
\end{center}

\numberwithin{definition}{section}
\numberwithin{theorem}{section}
\numberwithin{lemma}{section}
\numberwithin{proposition}{section}
\numberwithin{corollary}{section}
\numberwithin{example}{section}
\renewcommand{\theequation}{\thesection.\arabic{equation}}
\appendix

\section{Toolkit from graph theory}\label{appendix: graph}

A directed graph is a pair $G=(V,A)$, where $V$ is the finite set of vertices and $A\subseteq V\times V$ is the set of arcs.
We denote $(u,v)\in A$ to mean that the arc $(u,v)$ connects $u$ and $v$ in direction from $u$ to $v$. 
For each vertex $v\in V$ let $R(v)$ be the set of vertices that can be reached by directed paths starting from $v$. 
By definition, we have $v\in R(v)\subseteq V$ for any $v \in V$. 
The sets $(R(v))_{v \in V}$ induce an equivalence relation $\simeq$ on $V$ in such a way that, for each $u, v\in V$, $u\simeq v$ if and only if $u\in R(v)$ and $v\in R(u)$. 
Let $(U_i)_{i\in I}$ be the equivalence classes induced by $\simeq$. 
Since any $u,v \in U_i$ for each $i \in I$ is connected and any $u \in U_i, v \in U_j$ with $i \neq j$ is not connected, $H_i=(U_i, A|_{U_i})$ is a {\it strongly connected component} of $G$.
If $|I|=1$, we say that $G$ is {\it strongly connected}.

We can define a partial order $\preceq$ on the set  $\mathcal{H}(G)=\{H_i\mid i\in I\}$ such that $H_i\preceq H_j$ if and only if $u\in R(v)$ for every $u\in U_i$ and $v\in U_j$.
Let  $(\mathcal{H}(G),\preceq)$ be the induced partially ordered set (poset).
A strongly connected component $H_i$ is called a {\it maximal} strongly connected component of $G$ if there exists no strongly connected component $H_j\neq H_i$ such that $H_i\prec H_j$.
Let $\mathcal{H}^*(G) = \{H_i \in \mathcal{H}(G) \mid H_i \text{ is maximal}\}$ be the set of all maximal strongly connected components.
It is well-known that the decomposition of a graph $G=(V,A)$ into a collection of strongly connected components and that of maximal strongly connected components in linear time, i.e., ${\rm O}(|V|+|A|)$ time.

\begin{proposition}\label{prop: poset}
Suppose $G=(V,A)$ is a graph such that for every pair of distinct $u,v \in V$ there exists at least one of arcs $(u,v)$ and $(v,u)$ in $A$. Let $H_i = (U_i, B_i)(i \in I)$ be the strongly connected components of $G$ and $(\mathcal{H}(G),\preceq)$ be the induced poset. Then, the partial order $\preceq$ is a linear order.
\end{proposition}

\begin{proof}
Suppose to the contrary that there exists a non-comparable pair of distinct $H_i = (U_i, B_i)$ and $H'_i = (U_i', B'_i)$ in $H(G)$ with respect to $\preceq$. Then, for any $u \in U_i$ and $v \in U_i'$, there exists neither $(u,v) \in A$ nor $(v,u) \in A$, a contradiction.
\end{proof}

For any nonempty $U\subset V$, let $\Delta^+(U)$ be the set of all arcs $(u,v)\in A$ such that $u\in U$ and $v\in V\setminus U$. Also define $\Delta^-(U)$ to be the set of all arcs $(u,v)\in A$ such that $u\in V\setminus U$ and $v\in U$. 
A graph $G=(V,A)$ is called {\it connected} if for every nonempty proper subset $U\subset V$, at least one of $\Delta^+(U)$ and $\Delta^-(U)$ is nonempty.
If $\Delta^+(U)\neq\emptyset$ and $\Delta^-(U)=\emptyset$, we call $\Delta^+(U)$ a {\it directed cut} (or simply {\it dicut}) from $U$ to $V\setminus U$. 
Note that a connected graph is strongly connected if and only if there exists no dicut.
We can characterize maximal strongly connected components by the notion of dicut.

\begin{lemma}\label{lemma_maximal}
A strongly connected component $H_i = (U_i, B_i)$ is maximal if and only if $\Delta^-(U_i) = \emptyset$.
\end{lemma}

%%F (Delete)

\begin{lemma}\label{lemma_no-arcs}
Let $H_i, H_j \in \mathcal{H}^*(G)$ be distinct maximal strongly connected components.
Then there exists no arc between $U_i$ and $U_j$:
\[
    A \cap \bigl((U_i \times U_j) \cup (U_j \times U_i)\bigr) = \emptyset.
\]
\end{lemma}

\Delete{ %%F

\begin{proof}
If there were an arc $(u, v) \in A$ with $u \in U_i$ and $v \in U_j$, then $(u,v) \in \Delta^-(U_j)$, and we contradict Lemma~\ref{lemma_maximal}.
The case of an arc from $U_j$ to $U_i$ contradicts $\Delta^-(U_i) = \emptyset$ symmetrically.
\end{proof}

} %%F (Delete)

\section{Omitted proofs}\label{appendix:proof}

\begin{proof}[Proof of Proposition \ref{prop:schulze_mar}]
Notice that $F^{Schulze}(\mathcal{N})$ depends only on the relative size of $\tau^\gamma_{G_\gamma}(x,y)$ and $\tau^\gamma_{G_\gamma}(y,x)$, and $F^{Schulze}(\mathcal{N}_M)$ only on the relative size of $\tau^{\mu}_{G_\mu}(x,y)$ and $\tau^{\mu}_{G_\mu}(y,x)$ for all alternative pair $\{x,y\}$.

We first show $\tau^\gamma_{G_\gamma}(x,y) \ge \tau^\gamma_{G_\gamma}(y,x) \implies \tau^{\mu}_{G_\mu}(x,y) \ge \tau^{\mu}_{G_\mu}(y,x) \; \text{for all distinct} \; x,y \in X$.
We prove by the following two cases:\\

\noindent \underline{Case 1-1: $\max \{\tau^\gamma_{G_\gamma}(x,y) , \tau^\gamma_{G_\gamma}(y,x)\} > n/2$}.
Without loss of generality, suppose $\tau^\gamma_{G_\gamma}(x,y)\ge \tau^\gamma_{G_\gamma}(y,x)$.
Then, $\tau^{\gamma}_{G_\gamma}(x,y) > n/2$ implies $2\tau^{\gamma}_{G_\gamma}(x,y) -n = \tau^{\mu}_{G_\mu}(x,y) > 0$.
Since $\tau^{\gamma}_{G_\gamma}(y,x)$ is either less or equal than $n/2$ which results in $\Pi(y,x) =\emptyset$ in $\mathcal{N}_M$, or $n/2 < \tau^{\gamma}_{G_\gamma}(y,x) \le \tau^{\gamma}_{G_\gamma}(x,y)$, in both cases we have $\tau^{\mu}_{G_\mu}(x,y) \ge \tau^{\mu}_{G_\mu}(y,x)$.\\

\noindent \underline{Case 1-2: $\max \{\tau^\gamma_{G_\gamma}(x,y) , \tau^\gamma_{G_\gamma}(y,x)\} \le n/2$}.
Since for any alternative pair $\{x,y\}$ we must have either $\gamma(x,y) \ge n/2$ or $\gamma(y,x) \ge n/2$,
this implies $\tau^\gamma_{G_\gamma}(x,y) = \tau^\gamma_{G_\gamma}(y,x) = n/2$.
Then, $\Pi(x,y) = \Pi(y,x) = \emptyset$ in $\mathcal{N}_M$, and we have $\tau^\mu_{G_\mu}(x,y) = \tau^\mu_{G_\mu}(y,x) = 0$.
\\

Next, we show $\tau^\mu_{G_\mu}(x,y) \ge \tau^\mu_{G_\mu}(y,x) \implies \tau^\gamma_{G_\gamma}(x,y) \ge \tau^\gamma_{G_\gamma}(y,x)\; \text{for all distinct} \; x,y \in X$. We prove by the following two cases:\\

\noindent \underline{Case 2-1: $\max\{\tau^\mu_{G_\mu}(x,y), \tau^\mu_{G_\mu}(y,x)\} > 0$}.
Without loss of generality, suppose $\tau^\mu_{G_\mu}(x,y) \ge \tau^\mu_{G_\mu}(y,x)$.
This directly implies $\tau^\gamma_{G_\gamma}(x,y) \ge \tau^\gamma_{G_\gamma}(y,x)$.
\\

\noindent \underline{Case 2-2: $\max\{\tau^\mu_{G_\mu}(x,y), \tau^\mu_{G_\mu}(y,x)\} \le 0$}.
Since we either have $\gamma(x,y) \ge n/2$ or $\gamma(y,x) \ge n/2$, we must have $\gamma(x,y) = \gamma(y,x) = 0$ which implies $\tau^\gamma_{G_\gamma}(x,y) = \tau^\gamma_{G_\gamma}(y,x) = n/2$.
Thus, $\tau^\gamma_{G_\gamma}(x,y) \ge \tau^\gamma_{G_\gamma}$ holds.
\\

By the above-mentioned cases, we have shown
\[
\tau^\gamma_{G_\gamma}(x,y) \ge \tau^\gamma_{G_\gamma}(y,x) \iff \tau^\mu_{G_\mu}(x,y) \ge \tau^{\mu}_{G_\mu}(y,x) \; \text{for all distinct} \; x,y \in X,
\]
which implies that $\mathcal{N}$ and $\mathcal{N}_M$ preserve the relative size relation of the bottle-neck path weight.
Hence, we obtain $F^{Schulze}(\mathcal{N}) = F^{Schulze}(\mathcal{N}_M)$.
\end{proof}

\begin{proof}[Proof of Proposition~\ref{prop: SD_mar}]\label{proof:prop2}
Denote by $U^*$ the direct sum of the maximal strongly connected components of $\mathcal{N}_M$.
Then, Lemma~\ref{lemma_maximal} implies $\Delta^-(U^*) = \emptyset$.
Hence, for any $x \in U^*$ and $y \in V \backslash U^*$we have
\[
\tau^{\mu}_{G_\mu}(x,y) \ge \tau^{\mu}_{G_\mu}(y,x) = 0.
\]
Thus, the 
%%F schulze 
Schulze  %%F 
winner set $F^{Schulze}(\mathcal{N}_M)$ is a subset of $U^*$.
Moreover, since the schulze winner remains a schulze winner in an induced subgraph containing itself, applying the same logic in the proof of Theorem~\ref{thm:equivalence}, we obtain 
\begin{equation}\label{eq:mar}
F^{SD}(\mathcal{N}_M) = F^{Schulze}(\mathcal{N}_M).
\end{equation}
Combining to (\ref{eq:mar}) the results from Proposition~\ref{prop:schulze_mar} and Theorem~\ref{thm:equivalence} yields
\[
F^{SD}(\mathcal{N}) = F^{Schulze}(\mathcal{N}) = F^{Schulze}(\mathcal{N}_M) = F^{SD}(\mathcal{N_M}),
\]
which completes the proof.
\end{proof}

\bibliographystyle{ecta}
\bibliography{ref}

\end{document}